\documentclass[10pt, final, journal, letterpaper, twocolumn]{IEEEtran}

\usepackage{cite}
\usepackage{amsmath,amsthm,amssymb,amsfonts}
\usepackage{algorithm}
\usepackage{algorithmic}
\usepackage{graphicx}
\usepackage{color}
\usepackage[normalsize]{subfigure}

\newtheorem{theorem}{Theorem}

\begin{document}

\title{Optimal Mode Selection in D2D-Enabled Multi-Base Station Systems}
\author{Yuan~Liu,~\IEEEmembership{Member,~IEEE}
\thanks{Y. Liu is with the School of Electronic and Information Engineering, South China University of Technology, Guangzhou 510641, China (email: eeyliu@scut.edu.cn).}
}

\maketitle

\begin{abstract}
In this paper, we consider device-to-device (D2D) communication underlaying uplink cellular networks with multiple base stations (BSs), where each user can switch between traditional cellular mode (through BS) and D2D mode (by connecting proximity user), namely mode selection. We impose load balancing constraints on BSs to efficient resource usage. The joint problem of mode selection and user association is formulated as a combinatorial problem and NP-complete. We adopt a graph-based approach to solve the problem globally optimally in polynomial time. To further reduce complexity, we also propose a distributed algorithm based on dual method. We show that the proposed distributed algorithm achieves nearly the same performance as the proposed optimal graph based algorithm.
\end{abstract}


\section{Introduction}

The explosive growth of mobile users and their traffic requirements force researchers to seek new paradigms to revolutionize the traditional cellular networks. Device-to-device (D2D) communication recently has emerged as a promising technique for improving system performance. The idea is to enable two mobile devices in proximity of each other to establish a direct connection and to bypass the base station (BS). By incorporating D2D communication into cellular networks, mobile users can operate on \emph{D2D mode} (i.e, two users establish a direct local link) or traditional \emph{cellular mode} (i.e., two users communicate through BS), which may be controlled by either network or users. D2D communication gains many benefits, like spectrum/energy efficiency improvement, coverage extension, and cellular offloading.

However, D2D communication poses a set of new technical challenges, such as interference management and mode selection which are quite different from those of traditional cellular networks. There are a number of works for interference management by resource allocations \cite{Zulhasnine2010,Yu2011,Min2011,Xu2013,Feng2013,YuanICC2014}.
%
To take full advantage of D2D communication, each user needs to carefully switch between cellular mode and D2D mode. This is known as mode selection which is the unique problem in D2D-enabled cellular networks. In \cite{DopplerWCNC2010}, a mode selection procedure was proposed to limit interference caused by D2D communication.
In \cite{XiangGCW2012}, a distance-dependent mode selection algorithm with power optimization was analyzed. Joint mode selection and channel assignment based on stochastic framework was considered in \cite{Han2012}. Joint mode selection and power allocation was studied in \cite{Jung2012} by using exhaustive search all possible mode combinations of users. The authors in \cite{Chien2012} studied joint mode selection and resource allocation as as a mixed integer non-linear programming (MINLP) problem. Heuristic algorithms were proposed for joint mode selection, channel and power allocation in \cite{Guanding2014,Chenfei2014}.

In view of the related works based on mode selection \cite{DopplerWCNC2010,XiangGCW2012,Han2012,Jung2012,Chien2012,Guanding2014,Chenfei2014},
it is found that the previous efficient solutions for the resource allocation problems are always not optimal due to the combinational nature of mode selection. Additionally, most of the related works considered single-BS case. If with multiple BSs, the problem becomes more complex since it involves BS selection which tightly couples mode selection. Moreover, if with multiple BSs, load balancing is crucial for efficient and fair utilization of network resources, which also has not been considered in the D2D-enabled cellular networks yet.

Motivated by the above understanding, in this paper, we consider D2D communication underlaying \emph{uplink} cellular network assisted by multiple BSs. The considered system model is general and can be specified as a multicell scenario with multiple macro BSs or a small cell scenario where a macro BS coexists with multiple small BSs.
In the network, each user can switch between cellular mode and D2D mode. The main contributions of this paper are summarized as follows:
%
    We formulate the joint optimization problem of mode selection and BS selection in D2D-enabled multi-BS cellular networks. This joint problem has not been considered in the literature.
    The studied problem is a combinatorial problem and NP-complete. We adopt a centralized graph approach to solve the problem \emph{globally optimally} in polynomial time.
   To reduce the complexity and signaling overhead, we further propose a distributed algorithm which solves the original problem in the dual domain via Lagrangian dual decomposition. Simulation results show that the distributed algorithm performs very closely to the centralized algorithm.


\section{System Model and Problem Formulation}


We consider a D2D-enabled cellular network consisting of $M$ users and $N$ BSs. The considered multi-BS system model is general, which can be a multicell with multiple macro BSs or a small cell where a macro BS coexists with multiple small BSs (i.e., heterogeneous cellular networks). $\mathcal U=\{1,\cdots,M\}$ is the set of users and $\mathcal B=\{1,\cdots,N\}$ is the set of BSs. Each user consists of one dedicated transmitter and one dedicated receiver. In the following, we use ``user", ``transmitter", and ``receiver" interchangeably for convenience. Each user $i$  can associate with one BS by cellular mode or its own receiver by D2D mode. For notational convenience, we define $\mathcal B^+\triangleq\{0,1,\cdots,N\}$. That is to say, if a user's associated BS is $0$, the user does not connect any BS and operates on D2D mode. In this paper we focus on \emph{uplink} communication.
%
To simplify the problem, we do not consider power control and assume that the transmit power of users are fixed.

Let $g_{ij}$ denote the channel gain between user $i$'s transmitter and BS $j$, $h_{ij}$ the channel gain from transmitter $i$ to receiver $j$, $p_i$ the transmit power of transmitter $i$, and $n_0$ the power of the background noise. In this paper, we consider the Rayleigh block fading where the fading is assumed to occur identically (with Rayleigh distribution) and independently from one block to another but can be considered constant within each block.
Each user $i$'s performance is characterized by a utility function $u_i(\gamma_{ij})$, which is a function of the received signal-to-interference plus noise ratio (SINR) $\gamma_{ij}$ that depends on its association $\mathcal B^+$:
\begin{eqnarray}\label{eqn:sinr}
  \begin{cases}\gamma_{ij}=\frac{p_ig_{ij}}{\sum\limits_{k\in\mathcal U, k\neq i}p_kg_{kj}+n_0},~j\in\mathcal B,~{\rm for~ cellular~mode}\\
  \gamma_{ij}=\frac{p_ih_{ii}}{\sum\limits_{k\in\mathcal U, k\neq i}p_kh_{ki}+n_0},~j=0,~{\rm for~ D2D~mode}.\end{cases}
\end{eqnarray}
%

The utility $u_i(\gamma_{ij})$ can be defined as many popular metrics, such as rate $u_i(\gamma_{ij})=R_i(\gamma_{ij})$, weighted proportional fairness $u_i(\gamma_{ij})=w_i\log R_i(\gamma_{ij})$, and energy efficiency $u_i(\gamma_{ij})=\frac{R_i(\gamma_{ij})}{p_i}$, where $R_i$ is user $i$'s rate.
Since the transmit power of users are assumed to be fixed, all values of SINRs are deterministic through \eqref{eqn:sinr}, and so are the utilities. Thus different utility functions do not affect the proposed algorithms.


We denote $b_j$ as the maximum number of users that BS $j$ can support, or referred as to \emph{effective load}, and we assume that $b_j$ is a constant in each transmission slot. The operator adjusts $b_j$ to guarantee a certain degree of load balancing.
We let $\sum_{j\in\mathcal B}b_j=B$ and in general $B\leq M$ in D2D-enabled cellular network.

Each user can switch between cellular mode and D2D mode. Let $x_{ij}$ be the binary association variable with $x_{ij}=1$ indicating the connection of user $i$ to BS $j$ and $x_{ij}=0$ otherwise, $\forall i\in\mathcal U$, $\forall j\in\mathcal B^+$. Our goal is to determine the joint of user association and mode selection that maximizes the overall network utility. Mathematically, the problem can be formulated as
\begin{subequations}\label{eqn:p1}
\begin{align}
\textbf{P1}:~~\max_{\mathbf x}~&\sum_{i\in\mathcal U}\sum_{j\in\mathcal B^+}x_{ij}u_i(\gamma_{ij}) \\
{\rm s.t.}~~ &\sum_{j\in\mathcal B^+}x_{ij}\leq1,~~\forall i\in\mathcal U \label{eqn:x} \\
&\sum_{i\in\mathcal U}x_{ij}\leq b_j,~~\forall j\in\mathcal B \label{eqn:bj} \\
& x_{ij}\in\{0,1\},~~\forall i\in\mathcal U,\forall j\in\mathcal B^+, \label{eqn:binary}
\end{align}
\end{subequations}
where $\mathbf x\triangleq \{x_{ij}\}$; the constraint \eqref{eqn:x} states that each user can associate with at most one BS; the constraint \eqref{eqn:bj} ensures load balancing for BSs.


\section{Centralized Algorithm}

\textbf{P1} is an integer programming problem and NP-complete due to the combinational nature \cite{West}.
In this section, we propose a graph theoretical  approach to solve the problem \emph{globally optimally} in polynomial time.

Before beginning, we assume that a central controller is available so that the centralized resource allocation can be employed and the full network information can be perfectly gathered at the central controller. This is possible since the central controller can be embedded at a server in the core network for macrocells, e.g., radio network controller (RNC) which executes the task of resource management in universal mobile telecommunications system (UMTS).

Then, we show that \textbf{P1} is equivalent to a maximum weighted bipartite matching (MWBM) problem. To this end, we first briefly review some preliminaries of MWBM in graph theory \cite{West}.
%
   \emph{Bipartite graph:} A graph whose vertices are partitioned into two disjoint groups so that every edge connects a vertex in one group to one in another. If each edge is assigned a weight, the graph is the \emph{weighted bipartite graph}.
   \emph{Matching:} A set of mutually disjoint edges, i.e., any two edges do not share a common vertex. If every vertex in the graph is included in the matching, the matching is a \emph{perfect matching}.
   \emph{MWBM problem:} It is to find an optimal matching in a bipartite graph such that the sum weights of the matching is maximum.

\begin{theorem}\label{thm:graph}
The problem \textbf{P1} is equivalent to a MWBM problem.
\end{theorem}

\begin{proof}
Let $\mathcal G\triangleq(\mathcal V_1\times \mathcal V_2,\mathcal E,\mathcal W)$ be the weighted bipartite graph, where $\mathcal V_1$ and $\mathcal V_2$ are the two disjoint groups of vertices; $\mathcal E$ is the set of edges; $\mathcal W$ is the weighting function such that $\mathcal W:\mathcal E\rightarrow \mathbb{R}_+$. As depicted in Fig. \ref{fig:CL2015-2583-graph2}, we first let the transmitters be the vertices and put them in $\mathcal V_1$, and the receivers and BSs be the vertices in $\mathcal V_2$. In addition, we duplicate the vertex corresponding to each BS $j$ with $b_j$ identical vertices in $\mathcal V_2$. To make the graph balanced, some virtual vertices (the dotted circles in Fig. \ref{fig:CL2015-2583-graph2}) are added in $\mathcal V_1$ so that $|\mathcal V_1|=|\mathcal V_2|=M+B$, where $|\cdot|$ denotes the cardinality of a set. Each transmitter in $\mathcal V_1$ connects all possible (duplicated) BSs and its dedicated receiver in $\mathcal V_2$ by edges. Moreover, for each receiver in $\mathcal V_2$, we establish edges to connect it with the virtual vertices.   For every edge connecting the virtual vertices, its weight is zero.
That it, only the edges connecting each transmitter $i$ to the BSs and its dedicated receiver $j\in\mathcal B^+$ are assigned positive weights which are defined as the utility $u_i(\gamma_{ij})$. If a transmitter connects the duplicated $b_j$ BSs, the weights are the same as the connection to BS $j$.

Based on our graph construction method in above, we find the following equivalences: (i) The concept of a matching implies no violating the exclusive association constraints \eqref{eqn:x} and \eqref{eqn:binary}; (ii) The duplicated vertices in a matching satisfies the load balancing constraint \eqref{eqn:bj}; (iii) The weight of each edge is actually the utility of an association. Consequently, the problem in \textbf{P1} is exactly equivalent to the MWBM problem that is to find an optimal matching $\mathcal F^*\in\mathcal E$ so that the sum weights of $\mathcal F^*$ is maximum.
This completes the proof.
\end{proof}

\begin{figure}[t]
\begin{centering}
\includegraphics[scale=0.6]{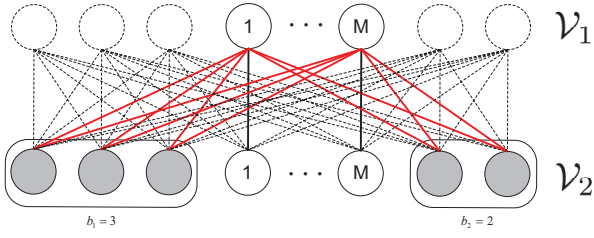}
\vspace{-0.1cm}
 \caption{An example of the bipartite graph with $M$ users and two BSs ($b_1=3$ and $b_2=2$). The weights of dash edges are zero.}\label{fig:CL2015-2583-graph2}
\end{centering}
\vspace{-0.3cm}
\end{figure}

After mapping \textbf{P1} to the MWBM problem, the well-known Hungary algorithm \cite{West} can solve the problem optimally with complexity of $\mathcal O((M+B)^{3})$, which is polynomial in the system parameters $M$ and $B$.

So far the proposed graph-based approach solved \textbf{P1} globally optimally, but it needs a central controller for centralized processing.
Nevertheless, the proposed centralized algorithm gives a performance upper bound and provides a simulation benchmark for the proposed distributed algorithms in next section.

\section{Distributed Algorithm}\label{sec:distributed}


In this section, we propose an efficiently suboptimal solution for \textbf{P1}, which works in a totally distributed manner. The idea is to solve the primal problem in dual domain via Lagrangian dual decomposition.

Let $\lambda_j$ denote the Lagrange multiplier corresponding to the load balancing constraint \eqref{eqn:bj} and $\boldsymbol\lambda\triangleq\{\lambda_j\}$. The Lagrangian function of of \textbf{P1} can be written as
\begin{equation}\label{eqn:La}
L(\mathbf x, \boldsymbol \lambda)=\sum_{i\in\mathcal U}\sum_{j\in\mathcal B^+}x_{ij}u_i(\gamma_{ij})+\sum_{j\in\mathcal B}\lambda_j\left(b_j-\sum_{i\in\mathcal U}x_{ij}\right).
\end{equation}
The dual function is then
\begin{eqnarray}\label{eqn:dual}
g(\boldsymbol\lambda)\triangleq\begin{cases}\max\limits_{\mathbf x}&L(\mathbf x, \boldsymbol\lambda)\\
{\rm s.t.} &\sum\limits_{j\in\mathcal B^+}x_{ij}\leq1,~~\forall i\in\mathcal U\\
&x_{ij}\in\{0,1\},~~\forall i\in\mathcal U,\forall j\in\mathcal B^+.
\end{cases}
\end{eqnarray}
The dual problem of \textbf{P1} is given by
%

\begin{eqnarray}\label{eqn:ming}
\begin{cases}  \min_{\boldsymbol\lambda}&g(\boldsymbol\lambda) \\
{\rm s.t.}~ &\boldsymbol\lambda\succcurlyeq0.
\end{cases}
\end{eqnarray}

Define $\lambda_0\equiv0$ for $j=0$ corresponding to D2D mode, maximizing $L(\mathbf x, \boldsymbol \lambda)$ is equivalent to maximizing the following objective:
\begin{equation}
\tilde{L}(\mathbf x, \boldsymbol \lambda)=\sum_{i\in\mathcal U}\sum_{j\in\mathcal B^+}x_{ij}u_i(\gamma_{ij})-\sum_{j\in\mathcal B^+}\lambda_j\sum_{i\in\mathcal U}x_{ij}.
\end{equation}
The above Lagrangian can be decomposed into $M$ sub-Lagrangian, and each sub-Lagrangian corresponds to one user and is expressed as
\begin{equation}\label{eqn:Li}
L_i(x_{ij})=\sum_{j\in\mathcal B^+}x_{ij}\left(u_i(\gamma_{ij})-\lambda_j\right).
\end{equation}
To maximize each sub-Lagrangian $L_i(x_{ij})$, the optimal $x_{ij}^*$ for each user $i$ can be determined by
\begin{eqnarray}\label{eqn:x-opt}
x_{ij}^*=\begin{cases}
  1,&{\rm if}~j=j^*=\arg\max\limits_{j\in\mathcal B^+} ~u_i(\gamma_{ij})-\lambda_j \\
  0,&{\rm otherwise}.
\end{cases}
\end{eqnarray}

The above analysis yields considerable insights as follows: The Lagrange multiplier $\lambda_j$, $j\in\mathcal B$, can be interpreted as the price announced by BS $j$ regarding on its loads. Then $u_i(\gamma_{ij})-\lambda_j$ can be treated as user $i$'s surplus, which is the difference between the utility by associating BS $j$ and its payment $\lambda_j$. If $j=0$ and $\lambda_0=0$, user $i$ operates on D2D mode and its surplus is just the utility (without payment), which coincides with the common sense of economics.

For the dual problem \eqref{eqn:ming}, the Lagrange multiplier $\lambda_j$ can be individually updated by each BS $j$ according to the following rule:
\begin{equation}\label{eqn:subgradient}
\lambda_j^{(t+1)}=\left[\lambda_j^{(t)}-\epsilon\left(b_j-\sum_{i\in\mathcal U}x_{ij}^{(t)}\right)\right]^+,~\forall j\in\mathcal B,
\end{equation}
where $[\cdot]^+\triangleq\max\{\cdot,0\}$ and $\epsilon$ is the properly designed step size. Since the dual problem in \eqref{eqn:dual} is always convex by definition, the subgradient method in \eqref{eqn:subgradient} is guaranteed to converge to the globally optimal solution to the dual problem \eqref{eqn:ming}.

In the distributed algorithm, every BS and receiver measure the SINRs and send to the transmitters via a feedback channel before the iteration. In each iteration, each user $i$ determines its association according to \eqref{eqn:x-opt} based on local information $\gamma_{ij}$ and the broadcasted information $\lambda_j$ for all $j$; Each BS updates its multiplier according to \eqref{eqn:subgradient} which only needs the local information, and then announces the new multiplier to the users. Thus the proposed dual-based algorithm can be implemented in a totally distributed fashion.
Note that the distributed algorithm incurs less signaling overhead compared to the proposed centralized algorithm. Specifically, each user only needs to send a beacon signal to its most favorable BS (if cellular mode is selected) rather than all BSs; at the BS side, only the price is needed to broadcast.


%

%

We note that due to the nonconvexity of \textbf{P1}, the optimal solution obtained in dual domain may not be the same as the primal optimum. This means that the so-called ``duality gap" exists, i.e., the difference between the optimal value of the primal
problem and that of the dual problem is non-zero. However, the proposed dual-based algorithm yields a good solution to the primal problem \textbf{P1} in some sense.

\section{Simulations}


We consider a two-dimensional plan of node locations, where the $N=4$ BSs are fixed at the coordinates $(0.25, 0.25)$, $(-0.25, 0.25)$, $(-0.25,-0.25)$, $(0.25,-0.25)$ of kilometer (km), and the $M=50$ users are uniformly distributed in the square with $1$ km of length of a side. The transmitter-receiver distance of each user is randomly distributed between $0$ and $0.2$ km. We set the path loss exponent for large-scale fading to be $4$ and the standard deviation of lognormal shadowing to be $5.8$ dB. The small-scale fading is modeled by Rayleigh fading process. A total of $2000$ different channel realizations with different locations and transmitter-receiver distances of users are used. Here we use the average throughput per user (total throughput divides the number of users) as the utility function for all users. Without loss of generality, we let the effective loads $b_j=b$ and the utility as the rate $u_i(\gamma_{ij})=R_i(\gamma_{ij})$ for all $i$ and $j$.

%

%
\begin{figure}[t]
\begin{centering}
\includegraphics[scale=.55]{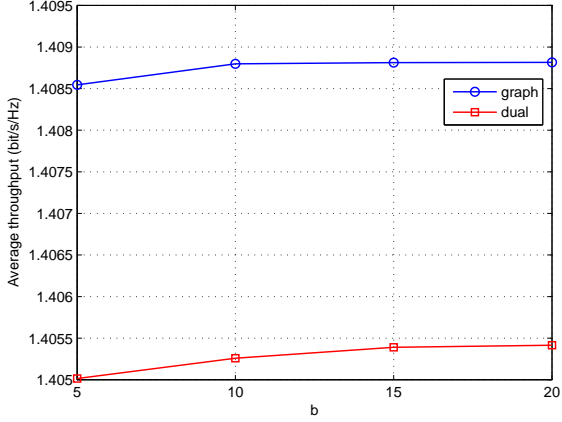}
\vspace{-0.1cm}
 \caption{Performance comparison of the proposed algorithms, where $p_i=-5$ dB for all $i$.}\label{fig:CL2015-2583-three200}
\end{centering}
\vspace{-0.3cm}
\end{figure}

We first evaluate the proposed  algorithms in Fig. \ref{fig:CL2015-2583-three200}. First, it is observed that the duality gap of the distributed algorithm really exists. However, the effect of the duality gap is negligible since the performance loss is less than $0.1\%$ compared with the globally optimal algorithm, which demonstrates the effectiveness of the distributed algorithm. Moreover, as expected, the throughput performance is increasing with the effective loads $b$ since more users can associate with the BSs.

%
\begin{figure}[t]
\begin{centering}
\includegraphics[scale=.55]{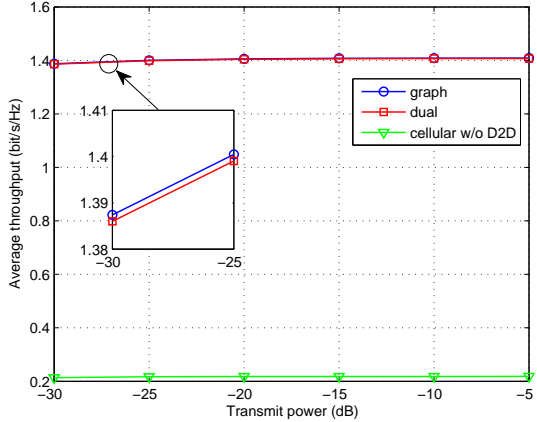}
\vspace{-0.1cm}
 \caption{Performance comparison of different algorithms with $b=10$.}\label{fig:CL2015-2583-rate_dB}
\end{centering}
\vspace{-0.3cm}
\end{figure}

Then we evaluate different algorithms in Fig. \ref{fig:CL2015-2583-rate_dB} where we set the effective loads of BSs as $b=10$. Here the traditional cellular communication without D2D mode is considered as the benchmark, which can be regarded as a special case of the proposed algorithms. We observe that average throughput is increasing with transmit power, but the performance will be bounded at the high transmit power due to the mutual interference. We also find that D2D mode dominates the performance in uplink transmission since the transmitter-receiver distances of users are usually very small and they prefer D2D mode compared to associating with BSs.





\section{Conclusion}

In this paper, we studied the optimal mode selection for D2D communications underlaying uplink cellular network with multiple BSs. To ensure fair and balanced utilization of resources, we imposed load balancing constraints in the BSs. We formulated the joint optimization problem of mode selection and user association as a combinatorial problem and NP-complete. We first solved this problem globally optimally in a centralized manner by using a graph based approach. To reduce complexity and signaling overhead, we further proposed a distributed algorithm which performs almost the same performance compared with the globally optimal one. Simulation results verified the proposed algorithms.

In this paper we consider the uplink case. It is also important and interesting to consider the problem in downlink case, which however needs to involve another dimensional resource allocation (like antennas, frequencies and/or slots). This will greatly complicate the problem and we would like to consider it for our future work.

\bibliographystyle{IEEEtran}
\bibliography{IEEEabrv,CL2015-2583-mode}
\end{document}